\documentclass[11pt,a4paper]{article}

\usepackage{amssymb,amsmath,amsthm}
\usepackage{graphicx}
\usepackage{color}
\usepackage{paralist}
\usepackage{titlesec}
\usepackage{algorithm}
\usepackage{algorithmic}
\usepackage{enumerate}
\usepackage{mathrsfs}

\graphicspath{{./}{figure/}}

\renewcommand{\leq}{\leqslant}
\renewcommand{\geq}{\geqslant}

\newcommand{\eps} {\varepsilon}

\newcommand{\cancel}[1]{\ignorespaces}

\newtheorem{theorem}{Theorem}[section]

\newtheorem{lemma}[theorem]{\bfseries{Lemma}}

\newbox\ProofSym
\setbox\ProofSym=\hbox{%
\unitlength=0.18ex%
\begin{picture}(10,10)
\put(0,0){\framebox(9,9){}}
\put(0,3){\framebox(6,6){}}
\end{picture}}

\begin{document}

	\title{Multistage Online Maxmin Allocation of \\ Indivisible Entities\thanks{Research supported by the
			Research Grants Council, Hong Kong, China (project no.~16207419).}}
	\author{Siu-Wing Cheng\footnote{Department of Computer Science and Engineering,
			HKUST, Hong Kong, China.}}
	
	\date{}
	
	\maketitle
	
	\begin{abstract}
	We consider an online allocation problem that involves a set $P$ of $n$ players and a set $E$ of $m$ indivisible entities over discrete time steps $1,2,\ldots,\tau$.  At each time step $t \in [1,\tau]$, for every entity $e \in E$, there is a restriction list $L_t(e)$ that prescribes the subset of players to whom $e$ can be assigned and a non-negative value $v_t(e,p)$ of $e$ to every player $p \in P$.  The sets $P$ and $E$ are fixed beforehand.  The sets $L_t(\cdot)$ and values $v_t(\cdot,\cdot)$ are given in an online fashion.   An allocation is a distribution of $E$ among $P$, and we are interested in the minimum total value of  the entities received by a player according to the allocation. In the static case, it is NP-hard to find an optimal allocation the maximizes this minimum value.  On the other hand, $\rho$-approximation algorithms have been developed for certain values of $\rho \in (0,1]$.  We propose a $w$-lookahead algorithm for the multistage online maxmin allocation problem for any fixed $w \geq 1$ in which the restriction lists and values of entities may change between time steps, and there is a fixed stability reward for an entity to be assigned to the same player from one time step to the next.  The objective is to maximize the sum of the minimum values and stability rewards over the time steps $1, 2, \ldots, \tau$.  Our algorithm achieves a competitive  ratio  of $(1-c)\rho$, where $c$ is the positive root of the equation $wc^2 = \rho (w+1)(1-c)$.  When $w = 1$, it is greater than $\frac{\rho}{4\rho+2} + \frac{\rho}{10}$, which improves upon the previous ratio of $\frac{\rho}{4\rho+2 - 2^{1-\tau}(2\rho+1)}$ obtained for the case of 1-lookahead.  
\end{abstract}

\section{Introduction}

Distributing a set $E$ of indivisible \emph{entities} to a set $P$ of \emph{players} is a very common optimization problem.  The problem can model an assignment of non-premptable computer jobs  to machines, a division of tasks among workers, allocating classrooms to lectures, etc.   The \emph{value} of an entity $e \in E$ to a player  $p \in P$ is usually measured by a non-negative real number.  In the single-shot case, the problem is to assign the entities to the players in order to optimize some function of the values of entities received by the players.  Every entity is assigned to at most one player.


The problem of maximizing the minimum total value of  entities assigned to a player is known as the \emph{maxmin fair allocation} or the \emph{Santa Claus} problem.  No polynomial-time algorithm can give an approximation ratio less than 2 unless P = NP~\cite{BD05}.  An LP relaxation of the problem, called configuration LP, has been developed; although its size is exponential,  it can be solved  by the  ellipsoid  method in polynomial time without an explicit construction of the entire LP~\cite{BS06}.  A polynomial-time algorithm was developed to round the fractional solution of the configuration LP to obtain an $\Omega(n^{-1/2}\log^{-3} n)$ approximation ratio~\cite{AS10}.  Subsequently, the approximation ratio was improved to $\Omega((n\log n)^{-1/2})$~\cite{SS}.  A tradeoff was obtained in~\cite{CCK09} between the approximation ratio and the exponent in the polynomial running time:~for any $\eps \geq 9\log\log n/\log n$, an $\Omega(n^{-\eps})$-approximate allocation can be computed in $n^{O(1/\eps)}$ time.  An important special case, the \emph{restricted} maxmin allocation problem, is that for every entity $e$, the value of $e$ is the same for players who want it and zero for the other players.  In this case, the configuration LP can  be used to give an~$\Omega(\log\log\log n/\log\log n)$-approximate allocation~\cite{BS06}. Later, it was shown that the  approximation ratio can be bounded by a large, unspecified constant~\cite{F,HSS}.  Subsequently, for any $\delta \in (0,1)$, the approximation ratio has been improved to $\frac{1}{6 + 2\sqrt{10} + \delta}$ in~\cite{AKS}, $\frac{1}{6+\delta}$ in~\cite{CM18,DRZ18}, and $\frac{1}{4+\delta}$ in~\cite{CM19,DRZ20}.

Recently, there has been interest in solving online optimization problems in a way that balances the optimality at each  time step and the stability of the solutions between successive time steps~\cite{BEM18,BCN14,BCNS12,CCDL16,GTW14}.  In the context of allocating indivisible entities, the following setting has been proposed in~\cite{BEM18}.  The sets of players and entities are fixed over a time horizon $t = 1, 2, \ldots, \tau$.  The value of $\tau$ may not be given in advance.  At the current time $t$, for every entity $e$, we are given the restriction list $L_t(e)$ of players to whom $e$ can be assigned and the value $v_t(e,p)$ of $e$ for every player $p \in P$.  We assume that $v_t(e,p) = 0$ if $p \not\in L_t(e)$.  In the strict online setting, no further information is provided.  In the $w$-lookahead  setting for any $w \geq 1$, we are given $L_{t+i}(\cdot)$ and $v_{t+i}(\cdot,\cdot)$ for every $i \in [0,w]$ at time $t$.  Note that $L_t(e) \subseteq P$; if $p \not= q$, $v_t(e,p)$ and $v_t(e,q)$ may be different; if $s \not= t$, $L_s(e)$ and $L_t(e)$ may  be different and so may $v_s(e,p)$ and $v_t(e,p)$.  At current time  $t$, we need to decide irrevocably an  allocation $A_t$ of the entities to the players so that the constraints given in $L_t(\cdot)$ are satisfied.  The objective is to maximize $\sum_{t=1}^\tau \min_{p \in P} \bigl\{\sum_{(e,p) \in A_t} v_t(e,p)\bigr\} + \sum_{t=1}^{\tau-1} \sum_{(e,p) \in E \times P} \Delta \cdot \bigl|\{(e,p) : (e,p) \in A_t \cap A_{t+1}\}\bigr|$, where $\Delta$ is some fixed non-negative value specified by the user.  The first term is the sum of the minimum total value of entities assigned to a player at each  time $t$.  A stability reward of $\Delta$ is given for keeping an entity at the same player between two successive time steps.  The second term is the sum of all stability rewards over all entities and all pairs of successive time steps.   The following results are obtained in~\cite{BEM18} for the multistage online maxmin allocation problem.  Let $\mathcal{A}$ be a $\rho$-approximation algorithm for some $\rho \leq 1$ for the single-shot maxmin allocation problem.  If $L_t(e) = P$ for every $e \in E$ and every $t \in [1,\tau]$, one can use $\mathcal{A}$ to obtain a competitive ratio of  $\frac{\rho}{\rho+1}$.  It takes $O(mn+T(m,n))$ time at each time step, where $T(m,n)$ denotes the running time of $\mathcal{A}$.  When the restriction lists $L_t(\cdot)$ are arbitrary subsets of $P$, it is impossible to achieve a bounded competitive ratio in the strict online setting.  On the other hand, using 1-lookahead, one can obtain a competitive ratio of $\frac{\rho}{4\rho + 2 - 2^{1-\tau}(2\rho+1)}$.  It takes $O(mn + T(m,n))$ time at each time step.

Two examples for the multistage online maxmin allocation problem are as follows.   Given a set of computing servers and some daily analytic tasks, the goal is to assign the executions of these tasks to the servers so that the minimum utilization of a server is maximized.  On each day, a task may only be executable at a particular subset of the servers due to resource requirements and data availability.  Moreover, there is a fixed gain in system efficiency by executing the same task at the same server on two successive days.  As the allocation of the daily analytic tasks to servers have to be performed quickly, one may choose $\mathcal{A}$ to be a polynomial-time approximation algorithm.  Nevertheless, in some planning problem, one may have enough time to solve the single-shot maxmin allocation problem exactly.  Consider an example in which a construction company is to produce assignments of engineers to different construction sites on an annual basis. Due to expertise and other considerations, an engineer can only work at a subset of the sites in the coming year.  The company wants to maximize the minimize annual progress of a site, and there is a fixed gain in efficiency in keeping an engineer at the same site from one year to the next.  If only a moderate number of engineers are involved, there may be enough time to take $\mathcal{A}$ to be an exact algorithm for solving the single-shot maxmin allocation problem.

In this paper, we improve the competitive ratio for the multistage online maxmin allocation problem and generalize to the case of $w$-lookahead for any fixed $w \geq1$.  We design a new online algorithm that achieves a competitive ratio of $(1-c)\rho$, where $c$ is the positive root of the equation $wc^2 = \rho(w+1)(1-c)$.  Our algorithm takes $O(wmn\log (wn) + w\cdot T(m,n))$ time at each time step.  The total time spent in invoking $\mathcal{A}$ for the entire time horizon $[1,\tau]$ is $O(\tau \cdot T(m,n))$.  If $w = 1$, our competitive ratio is greater than $\frac{\rho}{4\rho+2} + \frac{\rho}{10}$, which is better than the ratio of $\frac{\rho}{4\rho+2 - 2^{1-\tau}(2\rho+1)}$ for the case of 1-lookahead in~\cite{BEM18}.  

	
\section{Notation}

Let $I_t$ denote the input instance at time $t$ which specifies $L_t(\cdot)$ and $v_t(\cdot,\cdot)$.  Let $I_{a:b}$ denote the set of input instances $I_a, I_{a+1},\ldots, I_b$.  An  allocation  $C_t$ for $I_t$ is a set of ordered pairs $(e,p)$ for some $e \in E$ and $p \in P$ such that $p \in L_t(e)$ and every $e$ belongs to at most one pair  in $C_t$.  We use $C_{a:b}$ to denote the set of allocations $C_a,\ldots, C_b$ for the input instances $I_{a:b}$.  For every entity~$e$, $C_t[e]$ denotes the assignment of $e$ at time $t$ specified in $C_t$.  It is possible that $e$ is unassigned at time $t$.  For any interval $[a,b] \subseteq [1,\tau]$ and any entity $e$, we use $C_{a:b}[e]$ to denote the set of assignments $C_a[e],\ldots,C_b[e]$.

An alternative way to view $C_{1:\tau}$ is that it specifies a sequence of disjoint time intervals for every entity $e$.  In each time interval, $e$ is assigned to a single player.  We call these intervals \emph{assignment intervals}.  Our online algorithm generates a set of allocations $C_{1:\tau}$ by specifying these assignment intervals for the entities.  Because our algorithm does not know the all future instances, it is possible that it may generate two consecutive assignment intervals in which $e$ is assigned to the same player in them.  Ideally, we would like to merge such a pair of intervals; however, it is more convenient for our analysis to keep them separate.  Therefore, we do not assume that an assignment interval is a maximal interval such that $e$ is assigned to the same player, although this would be the case for the optimal offline solution for $I_{1:\tau}$.  

Given $C_{1:\tau}$ and any $[a,b] \subseteq [1,\tau]$, the assignment interval endpoints in $C_{a:b}$ refer to the endpoints of the assignment intervals in $C_{1:\tau}$ that lie in $[a,b]$.  The assignment intervals in $C_{a:b}$ refer to the assignment intervals in $C_{1:\tau}$ that are contained in $[a,b]$.  For every entity $e$, we can similarly interpret the notions of assignment interval endpoints in $C_{a:b}[e]$ and assignment intervals in $C_{a:b}[e]$.  Due to the constraints posed by $L_t(\cdot)$, it is possible that an entity $e$ is unassigned at some time step, so there may be a gap between an assignment interval end time and the next assignment interval start time in $C_{1:\tau}[e]$.  

Take any set of allocations $C_{1:\tau}$ for $I_{1:\tau}$.  
Define the following quantities:
\begin{align*}
\nu(C_t)  & = \min_{p \in P} \left\{\sum_{(e,p) \in C_t} v_t(e,p)\right\}, \\
\nu(C_{a:b}) & = \sum_{t=a}^b \nu(C_t), \\
\lambda(C_{t:t+1}[e]) & = \left\{\begin{array}{lcl}
	\Delta, & & \mbox{if $[t,t+1]$ is contained in an assignment interval of $e$}; \\
	0, & & \mbox{otherwise},
\end{array}\right. \\
\lambda(C_{a:b}[e]) & = \sum_{t=a}^{b-1} \lambda(C_{t:t+1}[e]), \\
\lambda(C_{a:b}) & = \sum_{e \in E} \lambda(C_{a:b}[e]).
\end{align*}
We call $\lambda(C_{a:b})$ the \emph{stability value} of $C_{a:b}$.  The value $\lambda(C_{t:t+1}[e])$ is \emph{stability reward of $e$ from $t$ to $t+1$}.

Our online algorithm requires a $w$-lookahead for any fixed $w \geq 1$.  That is, the input instances $I_{t+i}$ for all $i \in [0,w]$ are given at the current time step $t$.  
We assume that $I_{\tau+j}$ for any $j \geq 1$ is an empty input instance (i.e., instances in which $L_t(\cdot)$ are empty sets and $v_t(\cdot,\cdot)$ are zeros) so that we can talk about the $w$-lookahead at $\tau - i$ for any $i \in [0,w-1]$.


%

\section{Multistage online maxmin allocation}

\subsection{Overview and periods}

We start off by initializing a set $S_{1:\tau}$ of empty allocations.  Then, we use a greedy algorithm to update $S_{1:1+w}$ to be a set of allocations that maximize the stability value with respect to $I_{1:1+w}$.  We also use $S_{1:1+w}$ to generate the first \emph{period} as follows. 

The time step 1 is taken as a default \emph{period start time}.  In general, suppose that the current time step $s$ is a period start time.  Then, we use a greedy algorithm to compute the assignment intervals for some entities for $I_{s:s+w}$ provided by the $w$-lookahead.  This gives an updated $S_{s:s+w}$.  Every assignment interval end time in $S_{s:s+w}$ is a \emph{candidate end time}.  The time step $s+w$ is a default candidate end time.  For every assignment interval start time $i$ in $S_{s+1:s+w}$, $i-1$ is also a candidate end time.    (If $s$ is an assignment start time, $s$ does not induce $s-1$ as a candidate end time.)  Let $t$ be the smallest candidate end time within $[s,s+w]$.  Then, $[s,t]$ is the next period.  It is possible that $t = s+w$.  It is also possible that $t$ lies inside an assignment interval of an entity $e$ in $S_{1:s+w}[e]$.  To determine the allocations for $[s,t]$, we compute a set of allocations $B_{s:t}$ by running the $\rho$-approximation algorithm $\mathcal{A}$ on the instances $I_{s:t}$.    By a judicious comparison of $\nu(B_{s:t})$ and $\lambda(S_{s:t})$, we set the allocations $A_{s:t}$ to be $S_{s:t}$ or $B_{s:t}$.  $A_s$ will be returned at the current time step $s$; for each future time step $i \in [s+1,t]$,  $A_i$ will be returned.  The next period start time is $t+1$ and we will repeat the above at that time.

There are two main reasons for our improvement over the result in~\cite{BEM18}.   First, we do not recompute after some waiting time that is fixed beforehand.  The periods are dynamically generated and updated using the allocations produced by a greedy algorithm.  This fact allows us to make better use of the stability values offered by these greedy allocations.  Second, the greedy allocations and the $\rho$-approximate allocations are also compared in~\cite{BEM18} in determining the allocation for the current time step; however, our comparison is different because it allows us to reap the potential stability reward from the previous period to the current period, and at the same time, the potential stability reward from the current period to the next.

\subsection{Greedy allocations}

\begin{algorithm}[b]
	\begin{algorithmic}[1]
		\caption{StableAllocate$(a,b)$}
		\label{alg:2}
		\FOR{every entity $e$}
		\IF{no assignment interval in $S_{1:\tau}[e]$ starts before $a$ and contains $a$}
		\STATE{$S_{a:b}[e] \leftarrow \mathrm{StableEntity}(e,a,b)$}
		\ENDIF
		\ENDFOR
	\end{algorithmic}
\end{algorithm}

During the execution of our online algorithm, we maintain a set of allocations $S_{1:\tau}$ that are initially set to be empty allocations.  At any time step, the allocations in $S_{1:\tau}$ are possibly empty beyond some time in the future due to our limited knowledge of the future.   Suppose that $a$ is the start time of the next period.   Our online algorithm will call StableAllocate$(a,a+w)$, which is shown in Algorithm~\ref{alg:2}, to update $S_{a:a+w}$.   StableAllocate works by running a greedy algorithm for some of the entities.  Given an entity $e$ and an interval $[a,b]$, a greedy algorithm is described in~\cite{BEM18} to compute some assignments intervals of $e$ within $[a,b]$ that have the maximum stability value with respect to  the instances $I_{a:b}$.   We give the pseudocode, StableEntity, of this algorithm in Algorithm~\ref{alg:2-2}.   For every entity $e$, if no assignment interval in $S_{1:\tau}[e]$ starts before $a$ and contains $a$, we call StableEntity$(e,a,a+w)$ to recompute $S_{a:a+w}[e]$.  The updated allocations $S_{a:a+w}$ will serve two purposes.  First, they will determine the end time of the next period that starts at $a$.  Second, they will help us to determine the allocations that will be returned for the next period.

\begin{algorithm}[t]
	\begin{algorithmic}[1]
		\caption{StableEntity$(e,a,b)$}
		\label{alg:2-2}
		\STATE{initialize $C_{a:b}[e]$ to be empty allocations}
		\STATE{$i \leftarrow a$}
		\WHILE{$i \leq b$}
		\FOR{every player $q$}
		\IF{$q \in L_i(e)$}
		\STATE{$k_q \leftarrow \max\{ k \in [i,b] : q \in \bigcap_{j=i}^{k} L_j(e)\}$}
		\ELSE
		\STATE{$k_q \leftarrow 0$}
		\ENDIF
		\ENDFOR
		\STATE{$p \leftarrow \mathrm{argmax}\{k_q : q \in P \}$}
		\IF{$k_p \geq i$}
		\STATE{add $[i,k_p]$ as an assignment interval to $C_{a:b}[e]$ and assign $e$ to $p$ during $[i,k_p]$}
		\STATE{$i \leftarrow k_{p} + 1$}			
		\ELSE
		\STATE{$i \leftarrow i + 1$}
		\ENDIF
		\ENDWHILE
		\STATE{return $C_{a:b}[e]$}
	\end{algorithmic}
\end{algorithm}

The following result gives some structural conditions under which the stability value of some allocations $Y_{a:b}[e]$ is at least the stability value of some other allocations $X_{a:b}[e]$.  The greediness of StableEntity ensures that these conditions are satisfied by its output when compared with any $X_{a:b}[e]$.  As a result, Lemma~\ref{lem:1} proves the optimality of the stability value of the output of StableEntity.  The optimality of the greedy algorithm was also proved in~\cite{BEM18}, but we make the structural conditions more explicit in Lemma~\ref{lem:1}.

\begin{lemma}
	\label{lem:1}
	Let $[a,b]$ be any time interval.  Let $X_{a:b}[e]$ and $Y_{a:b}[e]$ be two sets of assignment intervals for $e$ that are contained in $[a,b]$.  If the following conditions hold, then for every $j \in [0,b-a]$, $\lambda(X_{a:a+j}[e]) \leq \lambda(Y_{a:a+j}[e])$.
	\begin{enumerate}[{\em (i)}]
		\item The first assignment interval in $Y_{a:b}[e]$ starts no later than the first assignment interval in $X_{a:b}[e]$.
		\item If the start time of an assignment interval $J$ in $Y_{a:b}[e]$ lies in an assignment interval $J'$ in $X_{a:b}[e]$, the end time of $J$ is not less than the end time of $J'$.
		\item For every $t \in [a,b]$, if $e$ is assigned in $X_t[e]$, then $e$ is also assigned in $Y_t[e]$.
	\end{enumerate}
\end{lemma}
\begin{proof}
	We show that $\lambda(X_{a:a+j}[e]) \leq \lambda(Y_{a:a+j}[e])$ for $j \in [0,b-a]$ by induction on $j$.  The base case of $j = 0$ is trivial as both $\lambda(X_{a:a}[e])$ and $\lambda(Y_{a:a}[e])$ are zero by definition.  Consider $a+j$ for some $j \in [1,b-a]$.  There are two cases depending the value of $\lambda(X_{a+j-1:a+j}[e])$.
	
	Case~1: $\lambda(X_{a+j-1:a+j}[e]) = 0$.  Then, $\lambda(X_{a:a+j}[e]) = \lambda(X_{a:a+j-1}[e]) \leq \lambda(Y_{a:a+j-1}[e]) \leq \lambda(Y_{a:a+j}[e])$.  
	
	Case~2: $\lambda(X_{a+j-1:a+j}[e]) = \Delta$.  Some assignment interval $J'$ in $X_{a:b}[e]$ contains $[a+j-1,a+j]$ in this case.  Let $p$ be the player to whom $e$ is assigned during $J'$.  Let $a+i$ be the start time of $J'$.  Note that $i \in [0,j-1]$.
	
	If $i = 0$, then $[a,a+j] \subseteq J'$ and $J'$ is the first assignment interval in $X_{a:b}[e]$.  By conditions~(i)~and~(ii), the first assignment interval in $Y_{a:b}[e]$ starts at $a$ and ends no earlier than $a+j$.  Therefore, $\lambda(X_{a:a+j}[e]) = \lambda(Y_{a:a+j}[e])$.  
	
	Suppose that $i > 0$.  
	Because $e$ is assigned to $p$ at $a+i$ in $X_{a:b}[e]$, by condition~(iii), there exists an assignment interval $J$ in $Y_{a:b}[e]$ that contains $a+i$.  So the start time of $J$ is less than or equal to $a+i$.  There are two cases.
	\begin{itemize}
		\item If the end time of $J$ is at least $a+j$, then $\lambda(X_{a+i:a+j}[e]) = \lambda(Y_{a+i:a+j}[e])$ and hence
		\begin{align*}
			\lambda(X_{a:a+j}[e]) & = \lambda(X_{a:a+i}[e]) + \lambda(X_{a+i:a+j}[e]) \\
			& \leq  \lambda(Y_{a:a+i}[e]) + \lambda(X_{a+i:a+j}[e])  
			& (\because \text{induction assumption}) \\
			& =  \lambda(Y_{a:a+i}[e]) + \lambda(Y_{a+i:a+j}[e])  \\
			& = \lambda(Y_{a:a+j}[e]).
		\end{align*}
	
		\item The other case is that $J$ ends at some time $t \in [a+i,a+j-1]$.  By condition~(ii), the start time of $J$ cannot be $a+i$, which means that the start time of $J$ is less than or equal to $a+i-1$.   Since $e$ is assigned to $p$ from $a+i$ to $a+j$ in $X_{a:b}[e]$, condition~(iii) implies that $e$ is assigned in $Y_{a:b}[e]$ at every time step in $[a+i,a+j]$.  Therefore, there is another assignment interval $K$ in $Y_{a:b}[e]$ that starts at $t+1$.  Condition~(ii) implies that the end time of $K$ is at least $a+j$.  There is thus a loss of a stability reward of $\Delta$ for $e$ from $t$ to $t+1$ in $Y_{a+i-1:a+j}[e]$, which matches the loss of a stability reward of $\Delta$ for $e$ from $a+i-1$ to $a+i$ in $X_{a+i-1:a+j}[e]$.  
		As a result, $\lambda(X_{a+i-1:a+j}[e])  = \lambda(Y_{a+i-1:a+j}[e])$.
		Hence,
		\begin{align*}
			\lambda(X_{a:a+j}[e]) & = \lambda(X_{a:a+i-1}[e]) + \lambda(X_{a+i-1:a+j}[e]) \\
			& \leq  \lambda(Y_{a:a+i-1}[e]) + \lambda(X_{a+i-1:a+j}[e])  
			& (\because \text{induction assumption}) \\
			& =  \lambda(Y_{a:a+i-1}[e]) + \lambda(Y_{a+i-1:a+j}[e])  \\
			& = \lambda(Y_{a:a+j}[e]).
		\end{align*}
	\end{itemize}
\end{proof}

\subsection{Online Algorithm}


The pseudocode of our online algorithm, MSMaxmin, is shown in Algorithm~\ref{alg:3}.  The parameter $c_0$ in line~\ref{code:3} is a real number chosen from the range $(0,1)$ that will be specified later when we analyze the performance of MSMaxmin.  

MSMaxmin initializes $S_{1:\tau}$ to be a set of empty allocations and then iteratively computes $A_1,A_2,\ldots,A_\tau$.   At the $s$-th time step, MSMaxmin calls StableAllocate$(s,s+w)$ if $s$ is the start time of the next period.   (By default, 1 is the start time of the first period.)   This call of StableAllocate updates $S_{1:\tau}$ by changing $S_{s:s+w}$.  Afterwards, we determine the period end time $t$ using the assignment interval start and end times in $S_{s:s+w}$.  The next task is to determine the allocations $A_{s:t}$ for $I_{s:t}$.

\begin{algorithm}[t]
	\begin{algorithmic}[1]
		\caption{MSMaxmin}
		\label{alg:3}
		\STATE{$S_{1:\tau} \leftarrow \text{empty allocations}$}
		\STATE{$\mathit{period}\_\mathit{start} \leftarrow 1$}
		\FOR{$s$ = 1 to $\tau$}
		\IF{$s = \mathit{period}\_\mathit{start}$}
		\IF{$s > 1$}
		\STATE{$[r,s-1] \leftarrow \mathit{period}$ \quad\quad\quad\quad\quad /* $[r,s-1]$ is the previous period */}
		\ENDIF
		\STATE{$\mathrm{StableAllocate}(s,s+w)$}
		\STATE{$t \leftarrow \min\bigl(s+w,\min\{\beta : \text{assignment interval end time $\beta$ in $S_{s:s+w}$}\}\bigr)$}    \label{code:1}
		\STATE{$t \leftarrow \min\bigl(t, \min\{\beta-1: \text{assignment start time $\beta$ in $S_{s+1:s+w}$}\}\bigr)$}  \label{code:1-1}
		\STATE{$\mathit{period} \leftarrow [s,t]$ \quad\quad\quad\quad\quad\quad\quad\quad /* $[s,t]$ is the next period */}
		\FOR{$j$ = $s$ to $t$}  \label{code:1-3}
		\STATE{$B_j\leftarrow \text{$\rho$-approximate maxmin allocation for $I_j$}$}
		\ENDFOR		\label{code:1-4}
		\STATE{$L \leftarrow 0$}   \label{code:2-1}
		\IF{$s > 1$ and $s \leq r+w$ and $A_{s-1} = S_{s-1}$}
		\STATE{$L \leftarrow \lambda(S_{s-1:s})$}
		\ENDIF
		\STATE{$R \leftarrow 0$}
		\IF{$t < s+w$}
		\STATE{$R \leftarrow \lambda(S_{t:t+1})$}
		\ENDIF
		\IF{$\nu(B_{s:t}) \geq L + \lambda(S_{s:t}) + c_0\cdot R$}  \label{code:3}
		\STATE{$A_{s:t} \leftarrow B_{s:t}$}
		\ELSE
		\STATE{$A_{s:t} \leftarrow S_{s:t}$}
		\ENDIF
		\STATE{$\mathit{period}\_\mathit{start} \leftarrow t+1$}  \label{code:2}
		\ENDIF
		\STATE{output $A_s$}
		\ENDFOR
	\end{algorithmic}
\end{algorithm}

We invoke the $\rho$-approximation algorithm for the single-shot maxmin allocation problem.  Specifically, we compute the $\rho$-approximate allocations $B_{s:t}$ for the instances $I_{s:t}$.  That is, for each $i \in [s,t]$, $\nu(B_i) \geq \rho \cdot\nu(X_i)$ for any allocation $X_i$ for $I_i$.  The $\rho$-approximation algorithm may not take the restriction lists into account.  Nevertheless, since we assume that $v_i(e,p) = 0$ if $p \not\in L_i(e)$, we can remove such an assignment $(e,p)$ from $B_i$ without affecting $\nu(B_i)$.  Therefore, we assume without loss of generality that every $B_i$ respects the restriction lists $L_i(\cdot)$.

We set $A_{s:t}$ to be $S_{s:t}$ or $B_{s:t}$.  It is natural to check whether $\lambda(S_{s:t})$ is larger than $\nu(B_{s:t})$.  However, if $s \leq r+w$ and $A_{s-1} = S_{s-1}$, where $r$ is the start time of the previous period, then $S_s$ was computed at time $r$ and it is possible that $\lambda(S_{s-1:s}[e]) = \Delta$ for some entity $e$.  The call StableAllocate$(s,s+w)$ does not invoke StableEntity for such an entity $e$, and so $S_s[e]$ will be preserved.  Therefore, if we set $A_{s:t}$ to be $S_{s,t}$, we will gain the stability reward of $\lambda(S_{s-1:s})$.  Similarly, if $t < s+w$, then setting $A_{s:t}$ to be $S_{s,t}$ provides the opportunity to gain the stability reward of $\lambda(S_{t:t+1})$ in the future.   On the other hand, if $t = s+w$, we do not know $I_{t+1}$ at time $s$, and we do not compute $S_{t+1}$ at time $s$.  In this case, after calling StableAllocate at time $s$, all assignment intervals in the current $S_{1:\tau}$ must end at or before $t = s+w$, implying that there is no stability reward in $S_{t:t+1}$ irrespective of how we will set $S_{t+1}$ in the future.  Hence, we compare $\nu(B_{s:t})$ with $\lambda(S_{s:t})$ and possibly $\lambda(S_{s-1:s})$ and $\lambda(S_{t:t+1})$ depending on the situation.

\section{Analysis}
 
Because MSMaxmin calls StableAllocate from time to time to update $S_{1:\tau}$, the set of allocations $S_{a:b}$ for any $[a,b] \subseteq [1,\tau]$ may change over time.  To differentiate these allocations computed at different times, we introduce the notation $S_{i|s}$ to denote $S_i$ at the end of the time step $s$, i.e., at the end of the $s$-th iteration of the for-loop in MSMaxmin.   Similarly, $S_{a:b|s}$ denotes $S_{a:b}$ at the end of the time step $s$.  

First, we give some properties of the start and end times of periods and assignment intervals.

\begin{lemma}
	\label{lem:period}
	Let $s$ be a period start time.  For every entity $e$,
	\begin{enumerate}[{\em (i)}]
		
		\item  if $\alpha$ is an assignment interval end time in $S_{s:s+w|s}[e]$, then $\alpha$ will remain an assignment interval end time for $e$ in the future and $\alpha$ will be a period end time; 
		
		\item if $\alpha$ is an assignment interval start time in $S_{s+1:s+w|s}[e]$, then $\alpha-1$ will be a period end time,  $\alpha$ will remain an assignment interval start time for $e$ in the future, and $\alpha$ will be a period start time.
		
	\end{enumerate}
\end{lemma}
\begin{proof}
	Take any entity $e$.  For any period start time $s$, the call StableAllocate$(s,s+w)$ invokes StableEntity$(e,s,s+w)$ only if $e$ is unassigned at time $s$ in $S_{1:\tau|s-1}[e]$ or $s$ is an assignment interval start time in $S_{1:\tau|s-1}[e]$.  Therefore, by the greediness of StableEntity, if a time step $\alpha\in [s,s+w]$ was already determined by some previous call of StableEntity on $e$ as an assignment interval start or end time, that decision will remain the same in this call StableEntity$(e,s,s+w)$.  Therefore, if $\alpha$ is an assignment interval end time for $e$ in $[s,s+w]$ after calling StableAllocate$(s,s+w)$, it will remain an assignment interval end time for $e$ in the future.  It will also be a candidate end time used in line~\ref{code:1} of MSMaxmin in the $s$-th iteration of the for-loop and thereafter until $\alpha$ becomes the next period end time.  We can similarly argue that if $\alpha$ is an assignment interval start time for $e$ in $[s,s+w]$ after calling StableAllocate$(s,s+w)$, it will remain an assignment interval start time for $e$ in the future.  Also, $\alpha-1$ will be a candidate end time used in line~\ref{code:1-1} of MSMaxmin until it becomes the next period end time.  When this happens, $\alpha$ will be made the subsequent period start time in line~\ref{code:2} of MSMaxmin.
\end{proof}

Next, we show that for every entity $e$, the stability value of any set of allocations $X_{1:t}[e]$ cannot be much larger than that of $S_{1:t|\alpha}[e]$, where $\alpha$ is the largest assignment interval start time in $S_{1:\tau|t}[e]$ that is less than or equal to $t$, provided that $t \leq \alpha+w$.

\begin{lemma}
	\label{lem:2}
	Let $X_{1:\tau}$ be a set of allocations for $I_{1:\tau}$.  Let $t$ be any time step.   Let $\alpha$ be the largest assignment interval start time in $S_{1:\tau|t}[e]$ that is less than or equal to $t$.  If $t \leq \alpha+w$, then $\lambda(X_{1:t}[e]) \leq \frac{w+1}{w}\lambda(S_{1:\alpha|\alpha}[e]) + \lambda(S_{\alpha:t|\alpha}[e])$.
\end{lemma}
\begin{proof}
	We prove the lemma by induction on the assignment interval start time $\alpha$.  In the base case, $\alpha$ is the smallest assignment interval start time for $e$ computed by StableEntity.  It follows from the greediness of StableEntity that $L_s(e) = \emptyset$ for all $s \in [1,\alpha-1]$, which implies that $X_{1:\alpha-1}[e]$ is empty.  By Lemma~\ref{lem:period}, $\alpha$ is a period start time, so MSMaxmin calls StableAllocate$(\alpha,\alpha+w)$ at $\alpha$ which calls StableEntity$(e,\alpha,\alpha+w)$.  By the greediness of StableEntity and Lemma~\ref{lem:1}, for every time step $t \in [\alpha,\alpha+w]$, $\lambda(X_{\alpha:t}[e]) \leq \lambda(S_{\alpha:t|\alpha}[e])$.  Hence, $\lambda(X_{1:t}[e]) = \lambda(X_{\alpha:t}[e]) \leq \lambda(S_{\alpha:t|\alpha}[e]) =\lambda(S_{1:t|\alpha}[e])$, i.e., the base case is true.
	
	Consider the induction step.  Let $\gamma$ be the largest assignment interval start time in $S_{1:\tau|t}[e]$ that is less than or equal to $t$.   To prove the lemma, we are only concerned with the case of $t \leq \gamma+w$.   
	By Lemma~\ref{lem:period}, $\gamma$ is a period start time, so MSMaxmin calls StableAllocate$(\gamma,\gamma+w)$ at $\gamma$ which calls StableEntity$(e,\gamma,\gamma+w)$.   By the greediness of StableEntity and Lemma~\ref{lem:1}, we get 
	\begin{equation}
		\lambda(X_{\gamma:t}[e]) \leq \lambda(S_{\gamma:t|\gamma}[e]).  \label{eq:2-4}
	\end{equation}
	Let $[\alpha,\beta]$ be the assignment interval in $S_{1:\tau|\gamma}[e]$ before $\gamma$.   Note that $\beta \leq \alpha+w$.  By Lemma~\ref{lem:period}, $\alpha$ is a period start time.  So MSMaxmin calls StableAllocate$(\alpha,\alpha+w)$ at $\alpha$ which calls StableEntity$(e,\alpha,\alpha+w)$.  The call StableEntity$(e,\gamma,\gamma+w)$ at $\gamma$ cannot modify assignment intervals that end before $\gamma$.  Also, as StableAllocate does not call StableEntity for $e$ within $[\alpha+1,\beta]$, $\alpha$ is the last time before $\gamma$ at which the assignment intervals for $e$ was updated by a call of StableEntity.  Therefore,
	\begin{equation}	
		S_{1:\beta|\alpha}[e] = S_{1:\beta|\gamma}[e].   \label{eq:2-2-1}
	\end{equation} 
	We claim that:
	\begin{equation}
		\forall\, s \in [\alpha,\gamma-1], \quad \lambda(X_{1:s}[e]) \leq \frac{w+1}{w}\lambda(S_{1:\alpha|\alpha}[e]) + \lambda(S_{\alpha:s|\alpha}[e]).
		\label{eq:2-1}
	\end{equation}
	For $s \in [\alpha,\beta]$, we have $\lambda(X_{1:s}[e]) \leq \frac{w+1}{w}\lambda(S_{1:\alpha|\alpha}[e]) + \lambda(S_{\alpha:s|\alpha}[e])$ by the induction assumption.  If $\beta+1\leq \gamma-1$, by the greediness of StableEntity, it must be the case that $L_s(e) = \emptyset$ for all $s \in [\beta+1,\gamma-1]$ so that StableEntity does not assign $e$ to any player during $[\beta+1,\gamma-1]$.  Therefore, $X_{\beta+1:s}[e]$ is empty for all $s \in [\beta+1,\gamma-1]$, which means that $\lambda(X_{1:s}[e]) = \lambda(X_{1:\beta}[e]) \leq \frac{w+1}{w}\lambda(S_{1:\alpha|\alpha}[e]) + \lambda(S_{\alpha:\beta|\alpha}[e]) =  \frac{w+1}{w}\lambda(S_{1:\alpha|\alpha}[e]) + \lambda(S_{\alpha:s|\alpha}[e])$.  Hence, \eqref{eq:2-1} holds.
	
	Suppose that $\gamma \leq \alpha+w$.   In this case, we know the instances $I_{\alpha:\gamma}$ at time $\alpha$.  By the greediness of StableEntity and Lemma~\ref{lem:1}, we have 
	\begin{equation}
		\lambda(X_{\alpha:\gamma}[e]) \leq \lambda(S_{\alpha:\gamma|\alpha}[e]).  \label{eq:2-2}
	\end{equation}
	As $\gamma \leq \alpha+w$, after calling StableEntity$(e,\alpha,\alpha+w)$ at $\alpha$, we already know that $\gamma$ is the start time of the next assignment interval for $e$.  
	Therefore, 
	there is no stability reward for $e$ from $\beta$ to $\gamma$ in $S_{\alpha:\gamma|\alpha}[e]$.
	Then, it follows from \eqref{eq:2-2-1} that 
	\begin{equation}
		\lambda(S_{\alpha:\gamma|\alpha}[e]) = \lambda(S_{\alpha:\gamma|\gamma}[e]).  \label{eq:2-3}
	\end{equation}
	Hence,
	\begin{align*}
		\lambda(X_{1:t}[e]) 
		& = \lambda(X_{1:\alpha}[e]) + \lambda(X_{\alpha:\gamma}[e]) + \lambda(X_{\gamma:t}[e]) \\
		& \leq \frac{w+1}{w}\lambda(S_{1:\alpha|\alpha}[e]) + \lambda(S_{\alpha:\gamma|\alpha}[e]) + \lambda(S_{\gamma:t|\gamma}[e])  & (\because \eqref{eq:2-4}, \eqref{eq:2-1}, \text{and}~\eqref{eq:2-2}) \\
		& = \frac{w+1}{w}\lambda(S_{1:\alpha|\gamma}[e]) + \lambda(S_{\alpha:\gamma|\gamma}[e]) + \lambda(S_{\gamma:t|\gamma}[e])   & (\because \eqref{eq:2-2-1}~\text{and}~\eqref{eq:2-3}) \\
		& \leq \frac{w+1}{w}\lambda(S_{1:\gamma|\gamma}[e]) + \lambda(S_{\gamma:t|\gamma}[e]).
	\end{align*}
	
	Suppose that $\gamma \geq \alpha+w+1$.  If $\beta < \alpha+w$, there is a gap $[\beta+1,\gamma-1]$ during which StableEntity does not assign $e$ to any player.  It means that $L_s(e) = \emptyset$ for all $s \in [\beta+1,\gamma-1]$.  Therefore, $e$ is also unassigned in $X_{\beta+1:\gamma-1}$ and we can conclude that
	\begin{align*}
		\lambda(X_{1:t}[e]) 
		& = \lambda(X_{1:\beta}[e]) +\lambda(X_{\gamma:t}[e]) \\
		& \leq \frac{w+1}{w}\lambda(S_{1:\alpha|\alpha}[e]) + \lambda(S_{\alpha:\beta|\alpha}[e]) + \lambda(S_{\gamma:t|\gamma}[e]) &  (\because \eqref{eq:2-4}~\text{and}~\eqref{eq:2-1}) \\
		& \leq \frac{w+1}{w}\lambda(S_{1:\beta|\gamma}[e]) + \lambda(S_{\gamma:t|\gamma}[e]) 
		& (\because \eqref{eq:2-2-1}) \\
		& = \frac{w+1}{w}\lambda(S_{1:\gamma|\gamma}[e]) + \lambda(S_{\gamma:t|\gamma}[e]).
		& (\because \forall s \in [\beta+1,\gamma-1],\, L_s(e) = \emptyset)
	\end{align*}
	The remaining case is that $\beta \geq \alpha+w$.  At time $\alpha$, MSMaxmin computes assignment intervals up to time $\alpha+w$ only.  It follows that $\beta = \alpha+w$, which implies that $\lambda(S_{\alpha:\beta|\alpha}[e]) = w\Delta$.  If the interval $[\beta+1,\gamma-1]$ is not empty, $e$ must be unassigned in $X_{\beta+1:\gamma-1}$ as we argued previously.  We can thus conclude as in the above that 	$\lambda(X_{1:t}[e]) \leq \frac{w+1}{w}\lambda(S_{1:\gamma|\gamma}[e]) + \lambda(S_{\gamma:t|\gamma}[e])$.  Suppose that $[\beta+1,\gamma-1]$ is empty.  It means that $\gamma = \beta+1$.  Then,
	\begin{align*}
		& \lambda(X_{1:t}[e]) \\
		& = \lambda(X_{1:\beta}[e]) + \lambda(X_{\beta:\beta+1}[e]) + \lambda(X_{\gamma:t}[e]) \\
		& \leq \frac{w+1}{w}\lambda(S_{1:\alpha|\alpha}[e]) + \lambda(S_{\alpha:\beta|\alpha}[e]) + \lambda(X_{\beta:\beta+1}[e]) + \lambda(S_{\gamma:t|\gamma}[e]) & (\because \eqref{eq:2-4}~\text{and}~\eqref{eq:2-1})\\
		& \leq \frac{w+1}{w}\lambda(S_{1:\alpha|\alpha}[e]) + \lambda(S_{\alpha:\beta|\alpha}[e]) + \Delta + \lambda(S_{\gamma:t|\gamma}[e]) 
		& (\because \lambda(X_{\beta:\beta+1}[e]) \leq \Delta) \\
		& = \frac{w+1}{w}\lambda(S_{1:\alpha|\alpha}[e]) + \frac{w+1}{w}\lambda(S_{\alpha:\beta|\alpha}[e]) + \lambda(S_{\gamma:t|\gamma}[e]) 
		& (\because \lambda(S_{\alpha:\beta|\alpha}[e]) = w\Delta) \\
		& = \frac{w+1}{w}\lambda(S_{1:\gamma|\gamma}[e]) + \lambda(S_{\gamma:t|\gamma}[e]).  
		& (\because \eqref{eq:2-2-1})
	\end{align*}
\end{proof}


We are ready to analyze the performance of MSMaxmin.  It depends on the parameter $c_0$ in line~\ref{code:3} of MSMaxmin which will be set based on the values of $\rho$ and $w$.

\begin{theorem}
	\label{thm:main}
		\emph{MSMaxmin} takes $O(wmn\log (wn) + w \cdot T(m,n))$ time at each period start time and $O(m)$ time at any other time step.  The total time taken by \emph{MSMaxmin} in running $\mathcal{A}$ for the entire time horizon $[1,\tau]$ is $O(\tau \cdot T(m,n))$.  Let $A_{1:\tau}$ be the solution returned by \emph{MSMaxmin}.  Then, $\lambda(A_{1:\tau}) + \nu(A_{1:\tau}) \geq \frac{wc_0^2}{w+1}\cdot\lambda(O_{1:\tau}) + (1-c_0)\rho \cdot \nu(O_{1:\tau})$, where $O_{1:\tau}$ is the optimal offline solution.  Hence, the competitive ratio is $(1-c_0)\rho$, where $c_0$ is the positive root of the equation $wc^2 = \rho (w+1)(1-c)$, that is, 
	\[
	c_0 = \frac{\sqrt{\rho^2(w+1)^2 + 4\rho w(w+1)} - \rho(w+1)}{2w}.
	\]
\end{theorem}
\begin{proof}
Let $s$ be the current time step.  If $s$ is not a period start time, MSMaxmin spends $O(m)$ time just to output $A_s$.  Suppose that $s$ is a period start time.   After calling StableAllocate$(s,s+w)$, we obtain $O(wm)$ assignment interval start and end times in $S_{s:s+w|s}$.  Selecting the next period end time in lines~\ref{code:1} and~\ref{code:1-1} can be done in $O(wm)$ time.  Running $\mathcal{A}$ in lines~\ref{code:1-3}--\ref{code:1-4} take $O(w \cdot T(m,n))$ time.  Lines~\ref{code:2-1}--\ref{code:2} clearly take $O(wm)$ time.  It remains to analyze the running time of the call StableAllocate$(s,s+w)$.  

Take an entity $e$ for which StableAllocate will call StableEntity.  We describe an efficient implementation of StableEntity as follows.  For every player $p$, we can construct the maximal interval(s) within $[s,s+w]$ in which $e$ can be assigned to $p$.  There are fewer than $w$ intervals for $p$.   We store these intervals for all players in a priority search tree $T_e$~\cite{mccrieght85}.  The tree $T_e$ uses $O(wn)$ space and can be organized in $O(wn\log (wn))$ time.  Given a time step~$i$, $T_e$ can be queried in $O(\log(wn))$ time to retrieve the interval that contains $i$ and has the largest right endpoint.  This capability is exactly what we need for determining $k_p$ in lines~4--11 of StableEntity.  It follows that the while-loop in StableEntity takes $O(w\log(wn))$ time.  The total running time of StableEntity for $e$ is thus $O(wn\log(wn))$, implying that StableAllocate takes $O(wmn\log(wn))$ time.  This completes the running time analysis.

We analyze the competitive ratio as follows.  Consider a period $[s,t]$.  Let $[r,s-1]$ be the period before $[s,t]$.  There are several cases.

\begin{itemize}
	
\item Case~1: $s > 1$ and $s \leq r+w$ and $A_{s-1} = S_{s-1|r}$. 

\begin{itemize}
	\item Case~1.1: $t < s+w$.   MSMaxmin compares $\lambda(S_{s-1:t|s}) + c_0\cdot\lambda(S_{t:t+1|s})$ with $\nu(B_{s:t})$. If $\lambda(S_{s-1:t|s}) + c_0\cdot\lambda(S_{t:t+1|s}) > \nu(B_{s:t})$, MSMaxmin sets $A_{s:t}$ to be $S_{s:t|s}$.  Because $s \leq r+w$, $A_{s} = S_{s|s}$ allows us to collect the stability reward of $\lambda(S_{s-1:s|s})$, which makes the contribution of $A_{s:t}$ to $\lambda(A_{1:\tau}) + \nu(A_{1:\tau})$ greater than or equal to
	\begin{equation} 
			\lambda(S_{s-1:t|s}) > (1-c_1)\cdot\lambda(S_{s-1:t|s})  - c_0c_1\cdot\lambda(S_{t:t+1|s}) + c_1\cdot\nu(B_{s:t}).    \label{eq:thm0}
	\end{equation}
	If $\lambda(S_{s-1:t|s}) + c_0\cdot\lambda(S_{t:t+1|s}) \leq \nu(B_{s:t})$, MSMaxmin sets $A_{s:t}$ to be $B_{s:t}$ and the contribution of $A_{s:t}$ to $\lambda(A_{1:\tau}) + \nu(A_{1:\tau})$ is at least 
	\begin{equation}
			\nu(B_{s:t}) \geq (1-c_1)\cdot\lambda(S_{s-1:t|s}) + c_0(1- c_1)\cdot\lambda(S_{t:t+1|s}) + c_1\cdot\nu(B_{s:t}).    \label{eq:thm1}
	\end{equation}

	\item Case~1.2: $t = s+w$.  In this case, MSMaxmin compares $\lambda(S_{s-1:t|s})$ with $\nu(B_{s:t})$.  If $\lambda(S_{s-1:t|s}) > \nu(B_{s:t})$, the contribution of $A_{s:t} = S_{s:t|s}$ to $\lambda(A_{1:\tau}) + \nu(A_{1:\tau})$ is greater than or equal to
		\begin{equation} 
		\lambda(S_{s-1:t|s}) > (1-c_1)\cdot\lambda(S_{s-1:t|s}) + c_1\cdot\nu(B_{s:t}).    \label{eq:thm2}
	\end{equation}
	If $\lambda(S_{s-1:t|s}) \leq \nu(B_{s:t})$, the contribution of $A_{s:t} = B_{s:t}$ is at least 
	\begin{equation}
		\nu(B_{s:t}) \geq (1-c_1)\cdot\lambda(S_{s-1:t|s}) + c_1\cdot\nu(B_{s:t}).    \label{eq:thm3}
	\end{equation}
	
\end{itemize}

\item Case~2: $s = 1$, or $s = r+w+1$, or $A_{s-1} = B_{s-1}$.  

	\begin{itemize}
		\item Case~2.1: $t < s+w$.   MSMaxmin compares $\lambda(S_{s:t|s}) + c_0\cdot\lambda(S_{t:t+1|s})$ with $\nu(B_{s:t})$.  If $\lambda(S_{s:t|s}) + c_0\cdot\lambda(S_{t:t+1|s}) > \nu(B_{s:t})$, the contribution of $A_{s:t} = S_{s:t|s}$ is greater than or equal to
		\begin{equation} 
			\lambda(S_{s:t|s}) > (1-c_1)\cdot\lambda(S_{s:t|s}) - c_0c_1\cdot\lambda(S_{t:t+1|s}) + c_1\cdot\nu(B_{s:t}).    \label{eq:thm4}
		\end{equation}
		If $\lambda(S_{s:t|s}) + c_0\cdot\lambda(S_{t:t+1|s}) \leq \nu(B_{s:t})$, the contribution of $A_{s:t} = B_{s:t}$ is at least 
		\begin{equation}
			\nu(B_{s:t}) \geq (1-c_1)\cdot\lambda(S_{s:t|s}) + c_0(1-c_1)\cdot\lambda(S_{t:t+1|s}) + c_1\cdot\nu(B_{s:t}).    \label{eq:thm5}
		\end{equation}

		\item Case~2.2: $t = s+w$.  In this case, MSMaxmin compares $\lambda(S_{s:t|s})$ with $\nu(B_{s:t})$.  If $\lambda(S_{s:t|s}) > \nu(B_{s:t})$, the contribution of $A_{s:t} = S_{s:t|s}$ is greater than or equal to
		\begin{equation} 
			\lambda(S_{s:t|s}) > (1-c_1)\cdot\lambda(S_{s:t|s}) + c_1\cdot\nu(B_{s:t}).    \label{eq:thm6}
		\end{equation}
		If $\lambda(S_{s:t|s}) \leq \nu(B_{s:t})$, the contribution of $A_{s:t} = B_{s:t}$ is at least 
		\begin{equation}
			\nu(B_{s:t}) \geq (1-c_1)\cdot\lambda(S_{s:t|s}) + c_1\cdot\nu(B_{s:t}).    \label{eq:thm7}
		\end{equation}
	\end{itemize}
\end{itemize}

Let $O_{1:\tau}$ be the optimal offline solution for $I_{1:\tau}$.

In the sum of the applications of \eqref{eq:thm0}--\eqref{eq:thm7} to all the periods, the $\nu(\cdot)$ terms sum to $c_1\cdot\nu(B_{1:\tau})$, which is at least $c_1\rho\cdot\nu(O_{1:\tau})$.

Consider the sum of the $\lambda(\cdot)$ terms.  Let $[r,s-1]$ and $[s,t]$ be two consecutive periods.  If \eqref{eq:thm6} or \eqref{eq:thm7} is applicable to $[r,s-1]$, then $s-1 = r+w$ and one of the inequalities \eqref{eq:thm4}--\eqref{eq:thm7} is applicable to $[s,t]$, implying that the sum of the $\lambda(\cdot)$ terms does not include $\lambda(S_{s-1:s|s})$.  Nevertheless, as $s-1 = r+w$, all assignment intervals computed at or before time $r$ do not extend beyond $s-1$.  Therefore, $\lambda(S_{s-1:s|s}) = 0$ and there is no harm done.  For all other kinds of transition from $s-1$ to $s$, the sum of the $\lambda(\cdot)$ terms includes the stability reward of the entities from $s-1$ to $s$ multiplied by a coefficient that is less than 1.  We analyze the smallest coefficient of the $\lambda(\cdot)$ terms as follows.

If  \eqref{eq:thm0} or \eqref{eq:thm4} is applicable to $[r,s-1]$, then $s-1 < r+w$ and we get a $-c_0c_1\cdot\lambda(S_{s-1:s|r})$ term.  In this case, one of the inequalities~\eqref{eq:thm0}--\eqref{eq:thm3} must be applicable to $[s,t]$, which contains the term $(1-c_1)\cdot\lambda(S_{s-1:s|s})$.  We claim that these two terms combine into $(1-c_1-c_0c_1)\cdot\lambda(S_{s-1:s|s})$.  Take any entity $e$.  If StableEntity is not invoked for $e$ at $s$, then $S_{s-1:s|r}[e] = S_{s-1:s|s}[e]$.  If StableEntity is invoked for $e$ at $s$, no assignment interval in $S_{r:r+w|r}[e]$ or $S_{r:r+w|s}[e]$ contains $[s-1,s]$ and so $\lambda(S_{s-1:s|r}[e]) = \lambda(S_{s-1:s|s}[e]) = 0$.  This proves our claim.

If  \eqref{eq:thm1} or \eqref{eq:thm5} is applicable to $[r,s-1]$, we get the term $c_0(1-c_1)\cdot \lambda(S_{s-1:s|r})$ which is equal to $c_0(1-c_1)\cdot\lambda(S_{s-1:s|s})$ as explained in the previous paragraph.

Among the coefficients of the $\lambda(\cdot)$ terms, the smallest ones are $1-c_1-c_0c_1$ and $c_0(1-c_1)$.  Balancing $1-c_1-c_0c_1$ and $c_0(1-c_1)$ gives the relation $c_0+c_1 = 1$.   As a result, $\lambda(A_{1:\tau}) + \nu(A_{1:\tau}) \geq c_0^2\cdot\lambda(S_{1:\tau|\tau}) + c_1\cdot\nu(B_{1:\tau}) \geq c_0^2\cdot\lambda(S_{1:\tau|\tau}) + (1-c_0)\rho \cdot \nu(O_{1:\tau})$.  Here, we use the fact that $S_{s:t|s}$ for a period $[s,t]$ will not be changed after $s$, and so $S_{s:t|s} = S_{s:t|\tau}$.  By Lemma~\ref{lem:2}, we get $\lambda(A_{1:\tau}) + \nu(A_{1:\tau}) \geq \frac{wc_0^2}{w+1} \cdot \lambda(O_{1:\tau|\tau}) + (1-c_0)\rho\cdot\nu(O_{1:\tau})$.  To maximize the competitive ratio, we balance the coefficients $\frac{wc_0^2}{w+1}$ and $(1-c_0)\rho$.  The only positive root of the quadratic equation $wc_0^2 = \rho(w+1)(1-c_0)$ is
\[
\frac{\sqrt{\rho^2(w+1)^2 + 4\rho w(w+1)} - \rho(w+1)}{2w}.
\]
This positive root is less than $\bigl((\rho(w+1) + 2w) - \rho(w+1)\bigr)/(2w) = 1$.
\end{proof}



Suppose that we keep $\rho$ general and set $w = 1$.  Then, $c_0 = \sqrt{\rho^2 + 2\rho} - \rho$ and our competitive ratio is $(\rho+1-\sqrt{\rho^2+2\rho})\rho$.  To compare with the $\frac{\rho}{4\rho+2}$ bound in~\cite{BEM18}, we consider the difference in the coefficients $\rho+1 - \sqrt{\rho^2+2\rho} - 1/(4\rho+2)$.  Treating this as a function in $\rho$, the derivative of this difference is $1 - (\rho+1)(\rho^2+2\rho)^{-1/2} + (2\rho+1)^{-2}$.  This derivative is negative for $\rho \in (0,1]$, so the smallest difference is roughly $0.2679 - 0.1667 > 0.1$ when $\rho=1$.  Therefore, our competitive ratio is greater than $\frac{\rho}{4\rho+2} + \frac{\rho}{10}$.


\section{Conclusion}

We presented a $w$-lookahead online algorithm for the multistage online maxmin allocation problem for any fixed $w \geq 1$.   It is more general than the 1-lookahead online algorithm in the literature~\cite{BEM18}.  For the case of $w=1$, our competitive ratio is greater than $\frac{\rho}{4\rho+2} + \frac{\rho}{10}$, which improves upon the previous ratio of $\frac{\rho}{4\rho+2 - 2^{1-\tau}(2\rho+1)}$ in~\cite{BEM18}.  It is unclear whether our analysis of MSMaxmin is tight.  When we set $A_{s:t}$ to be $S_{s:t}$, we only analyze $\lambda(S_{s:t})$ and ignore $\nu(S_{s:t})$.  Conversely, when we set $A_{s:t}$ to be $B_{s:t}$, we only analyze $\nu(B_{s:t})$ and ignore $\lambda(B_{s:t})$.   There may be some opportunities for improvement.

\bibliographystyle{plain}
\bibliography{ref}

\begin{thebibliography}{10}

\bibitem{AKS}
C.~Annamalai, C.~Kalaitzis, and O.~Svensson.
\newblock Combinatorial algorithms for restricted max-min fair allocation.
\newblock {\em ACM Transactions on Algorithms}, 13:article~37, 2017.

\bibitem{AS10}
A.~Asadpour and A.~Saberi.
\newblock An approximation algorithm for max-min fair allocation of indivisible
  goods.
\newblock {\em SIAM Journal on Computing}, 39:2970--2989, 2010.

\bibitem{BEM18}
E.~Bampis, B.~Escoffier, and S.~Mladenovic.
\newblock Fair resource allocation over time.
\newblock In {\em Proceedings of the 17th International Conference on
  Autonomous Agents and Multiagent Systems}, pages 766--773, 2018.

\bibitem{BS06}
N.~Bansal and M.~Sviridenko.
\newblock The {S}anta {C}laus problem.
\newblock In {\em Proceedings of the Annual ACM Symposium on Theory of
  Computing}, pages 31--40, 2006.

\bibitem{BD05}
I.~Bez\'{a}kov\'{a} and V.~Dani.
\newblock Allocating indivisible goods.
\newblock {\em ACMSIGecom Exchange}, 5:11--18, 2005.

\bibitem{BCN14}
N.~Buchbinder, S.~Chen, and J.S. Naor.
\newblock Competitive analysis via regularization.
\newblock In {\em Proceedings of the ACM-SIAM Symposium on Discrete
  Algorithms}, pages 436--444, 2014.

\bibitem{BCNS12}
N.~Buchbinder, S.~Chen, J.S. Naor, and O.~Shamir.
\newblock Unified algorithms for online learning and competitive analysis.
\newblock In {\em Proceedings of the 25th Annual Conference on Learning
  Theory}, volume~23, pages 5.1--5.18, 2012.

\bibitem{CCK09}
D.~Chakrabarty, J.~Chuzhoy, and S.~Khanna.
\newblock On allocating goods to maximize fairness.
\newblock In {\em Proceedings of the 50th IEEE Symposium on Foundations of
  Computer Science}, pages 107--116, 2009.

\bibitem{CM18}
S.-W. Cheng and Y.~Mao.
\newblock Restricted max-min fair allocation.
\newblock In {\em Proceedings of the 45th International Colloquium on Automata,
  Languages, and Programming}, pages 37:1--37:13, 2018.

\bibitem{CM19}
S.-W. Cheng and Y.~Mao.
\newblock Integrality gap of the configuration lp for the restricted max-min
  fair allocation.
\newblock In {\em Proceedings of the 46th International Colloquium on Automata,
  Languages, and Programming}, pages 38:1--38:13, 2019.

\bibitem{CCDL16}
E.~Cohen, G.~Cormode, N.~Duffield, and C.~Lund.
\newblock On the tradeoff between stability and fit.
\newblock {\em ACM Transactions on Algorithms}, 13(1):7:1--7:24, 2016.

\bibitem{DRZ18}
S.~Davies, T.~Rothvoss, and Y.~Zhang.
\newblock A tale of santa claus, hypergraphs and matroids.
\newblock {\em CoRR}, abs/1807.07189, 2018.

\bibitem{DRZ20}
S.~Davies, T.~Tothvoss, and Y.~Zhang.
\newblock A tale of {Santa Claus}, hypergraphs and matroids.
\newblock In {\em Proceedings of the 31st Annual ACM-SIAM Symposium on Discrete
  Algorithms}, pages 2748--2757, 2020.

\bibitem{F}
U.~Feige.
\newblock On allocations that maximize fairness.
\newblock In {\em Proceedings of the 19-th Annual ACM-SIAM Symposium on
  Discrete Algorithms}, pages 287--293, 2008.

\bibitem{GTW14}
A.~Gupta, K.~Talwar, and U.~Wieder.
\newblock Changing bases: multistage optimization for matroids and matchings.
\newblock In {\em Proceedings of the International Colloquium on Automata,
  Languages, and Programming}, pages 563--575, 2014.

\bibitem{HSS}
B.~Haeupler, B.~Saha, and A.~Srinivasan.
\newblock New constructive aspects of the lov\'{a}sz local lemma.
\newblock {\em Journal of the ACM}, 58:article no.~28, 2011.

\bibitem{mccrieght85}
E.M. McCrieight.
\newblock Priority search trees.
\newblock {\em SIAM Journal on Computing}, 14(2), 1985.

\bibitem{SS}
B.~Saha and A.~Srinivasan.
\newblock A new approximation technique for resource-allocation problems.
\newblock {\em Random Structures and Algorithms}, 52:680--715, 2018.

\end{thebibliography}

\end{document}